\documentclass[11pt]{article}

\usepackage[margin=1in]{geometry}
\usepackage[T1]{fontenc}
\usepackage{amsmath,amssymb,amsthm}
\usepackage{booktabs,microtype,tabularx}
\usepackage{graphicx}
\usepackage{placeins}
\graphicspath{{./}}
\newcommand{\E}{\mathbb{E}}

\newcommand{\Q}{\mathbb{Q}}
\newcommand{\cD}{\mathcal{D}}
\newcommand{\cL}{\mathcal{L}}
\newcommand{\cR}{\mathcal{R}}
\newcommand{\dd}{\mathrm d}
\newcommand{\norm}[1]{\lVert#1\rVert}
\newcommand{\pos}[1]{[#1]_+}
\newtheorem{maintheorem}{Theorem}
\newtheorem{mainproposition}{Proposition}
\title{Amortizing the Calibration Triple: A Projection-Consistent Neural
Operator for Local-Stochastic Volatility}
\author{Xiaozhen Wang$^{1}$, Ana\"is Despr\'es$^{2,3}$, Martin Dureau$^{3}$,
Francois Buet-Golfouse$^{3,*}$\\[0.6em]
\small $^{1}$CEREMADE, Universit\'e Paris Dauphine-PSL\\
\small $^{2}$LaMME, Universit\'e \'Evry Paris-Saclay\\
\small $^{3}$AIML Global Markets, Barclays\\[0.4em]
\small $^{*}$Corresponding author: \texttt{francois.buetgolfouse@barclays.com}}
\date{}

\begin{document}

\maketitle

\begin{abstract}
Local-stochastic volatility (LSV) combines vanilla marginals with richer smile dynamics, but calibration requires a slow, noisy and sequential McKean--Vlasov fixed point. We learn a projection-consistent operator for the calibration triple. Given finite quotes and a stochastic-volatility (SV) backbone, it jointly returns an implied-volatility surface subject to static-arbitrage constraints, its Dupire local volatility, LSV leverage and the conditional moment required by the projection identity. Starting from option-price marginals, we derive a division-free Dupire residual in log-implied-variance coordinates and a quotient Fokker--Planck equation after Gy\"ongy projection. Deep Operator Network (DeepONet) and Fourier Neural Operator (FNO) implementations enforce quote fit, static-arbitrage, Dupire and projection constraints. For the witness-augmented residual system, we prove conditional identification and empirical consistency under LSV existence and inverse residual stability. In controlled synthetic tests, forward-start and cliquet errors differ from a particle method by 0.1 and 0.2 percentage points, while calibration latency falls from 98.5 to 0.6 ms. Compared with the tested baselines, local-volatility root-mean-square error (RMSE) falls by 36\% and leverage RMSE by 7--16\%. These results support amortizing the LSV fixed point: the expensive solve moves offline, while online calibration reduces to a single projection-consistent operator evaluation.
\end{abstract}

\section{Introduction}
LSV calibration is usually carried out in three steps. Starting from a finite
bid--ask quote book, a desk fits an arbitrage-admissible implied-volatility
(IV) surface, differentiates it to obtain Dupire local variance, and then
solves for LSV leverage. SSVI is one way to interpolate the quotes
\cite{gatheral2014}; Andreasen--Huge is a standard local-volatility convention
\cite{andreasen2011}; and particle methods solve the LSV projection when a
calibrated model exists \cite{guyon2012,guyon2013,reisinger2025}. The Dupire
step is sensitive to sparse quotes and the smoothing convention because it
differentiates the fitted surface. The leverage step is instead a
McKean--Vlasov fixed point: its conditional moment is computed under dynamics that already
contain the unknown leverage, which makes the calibration a McKean--Vlasov
problem. Particle estimates add Monte Carlo and kernel error, and the
solve proceeds sequentially in time. These numerical errors also pass from one
stage to the next: local variance inherits the quote-smoothing and
differentiation error, while leverage inherits both that field and the noisy
conditional-moment estimate. Because the stages are calibrated by separate
engines, their outputs need not form a projection-consistent triple. The whole
calculation must still be repeated for every surface, backbone and scenario.

We address both problems by learning a projection-consistent conditional
operator and moving the repeated solve offline. For quote book $Q$ and
stochastic-volatility (SV) backbone descriptor $\theta$, the operator returns
\begin{equation}
 \mathcal G^\star:(Q,\theta)\longmapsto
 (I,b=I^2,q=\log b,a_D,m_\theta,\ell_\theta),                 \label{eq:operator}
\end{equation}
where $a_D$ is Dupire local variance,
$m_\theta(t,x)=\E[V_\theta(Y_t)\mid X_t=x]$, and
$\ell_\theta^2m_\theta=a_D$. The operator returns the calibration triple
(IV, Dupire local volatility and LSV leverage) together with the conditional
moment linking its last two components. At inference, residual gates decide
whether the output is accepted, passed to a short particle polish, or rejected.

The output fields are tied together by the projection identities.
Option prices determine marginal laws; Dupire identifies their
one-dimensional generator; the selected backbone lifts a marginal into a
joint-law fibre; and Markovian projection selects leverage within that fibre
\cite{gyongy1986}. Substituting the projection quotient into the joint
Fokker--Planck equation exposes the nonlinear fixed point to the learning
problem. We use the same identities to couple the output heads during training
and to audit their predictions after inference.

\begin{figure}[htbp]
\centering
\includegraphics[width=0.8\columnwidth]{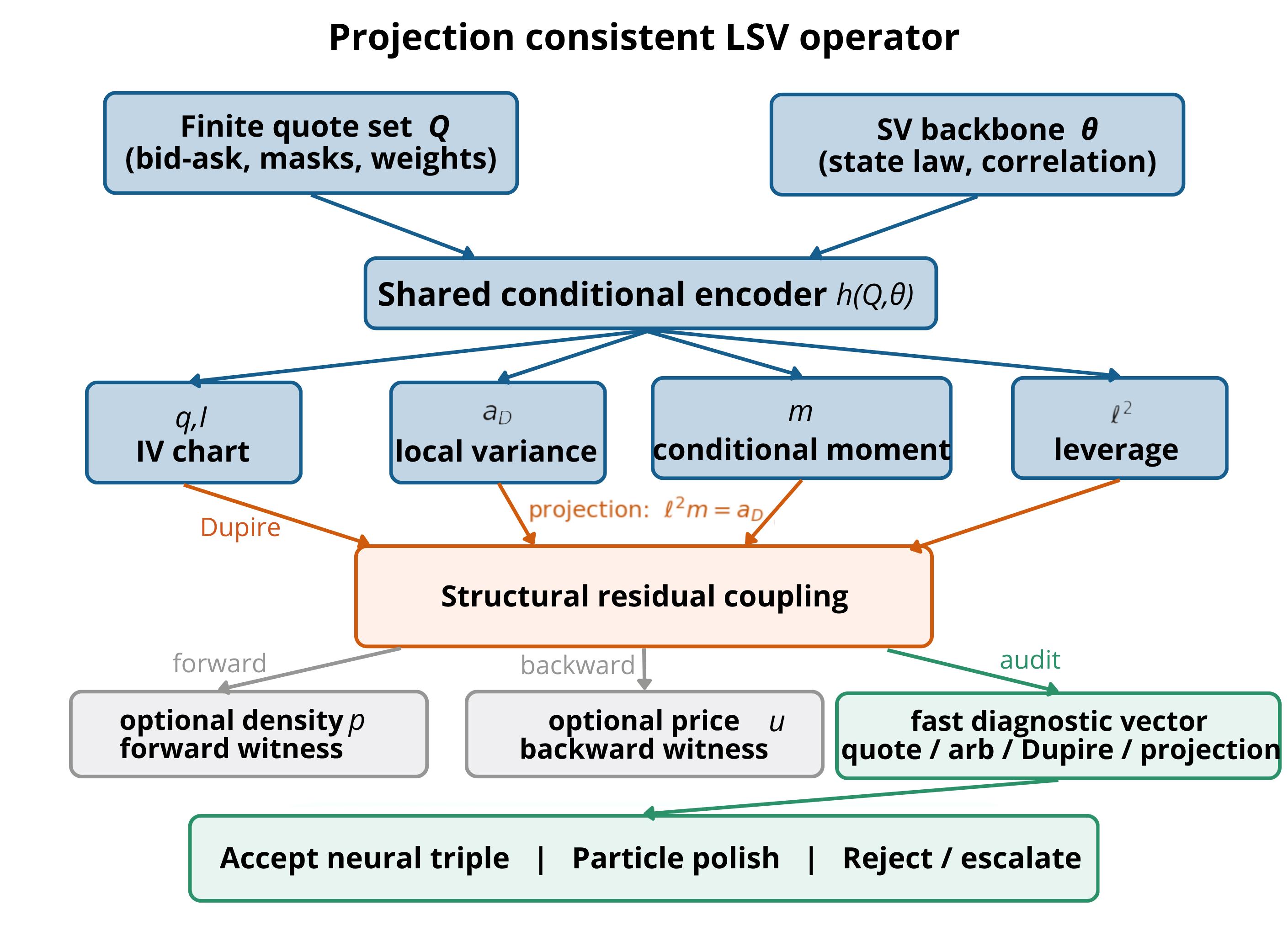}
\caption{Architecture and audit flow. The blue fast path maps finite quotes and
the selected SV backbone to the calibration triple. Orange residuals couple the
heads; gray optional witnesses define the full identifying audit extension.
Green diagnostics route the result to acceptance, particle polish or rejection.}
\label{fig:architecture}
\end{figure}

The paper makes four contributions. First, it derives a division-free Dupire
residual in log-IV coordinates and the quotient McKean--Vlasov equation directly
from normalized calls. Second, it implements conditional DeepONet and FNO
models whose chart and projection heads are coupled through structural
residuals; optional forward/backward
witnesses yield an auditable identifying system and a conditional consistency
theorem (Fig.~\ref{fig:architecture}). Third, it proves that finite quotes do
not identify a unique Dupire field, motivating residual-gated, annealed use of
heterogeneous desk calibrators rather than a single exact teacher.
Fourth, it separates online calibration latency from downstream pricing and
evaluates methods that produce LSV dynamics. The controlled synthetic
experiments support learning the fixed-point solve offline while retaining
residual checks and optional particle polish online. Historical-market and
production validation are outside the scope of this study.

\FloatBarrier
\section{Related work}
Volatility-surface construction and LSV calibration have largely developed as
separate layers. Arbitrage-free smoothing, SVI/SSVI and Andreasen--Huge
interpolation construct IV or local-volatility surfaces from quotes
\cite{fengler2009,gatheral2014,andreasen2011}, while neural smoothers amortize
the IV stage \cite{ackerer2020,wiedemann2024,yang2025}. Dupire recovers local
volatility from the fitted surface \cite{dupire1994}. For LSV models,
Gy\"ongy's projection relates local variance to a conditional
stochastic-variance moment \cite{gyongy1986}; particle schemes then solve for
leverage \cite{guyon2012,guyon2013,cozma2019,reisinger2025}. Existence and
stability are known mainly in regularized or restricted settings
\cite{abergel2010,lacker2020,bayer2024}.

Learning-based calibration has mainly targeted parameter inversion or a
single model and surface. Deep calibration learns parameter-to-price or
parameter-to-IV maps \cite{bayer2018,horvath2021}, while neural LSV methods
fit one surface or replace the conditional-expectation regression inside a
classical solver \cite{cuchiero2020,hakala2019}. Differential learning uses
derivative supervision \cite{huge2020}; DeepONet and FNO provide architectures
for maps between functions \cite{lu2021,li2021,kovachki2023}. Less attention
has been paid to an amortized map that spans the IV chart, Dupire field,
conditional moment and leverage while preserving the Dupire and Markovian
projection relations.

\section{From quotes to the nonlinear projection}
\subsection{Marginals and Dupire in IV coordinates}
Let $F_0(T)$ and $P_{0T}$ denote the forward and discount factor,
$X_t=\log(S_t/F_0(t))$, $k=\log(K/F_0(T))$, and
\begin{equation}
 c(T,k)=\frac{C(T,F_0(T)e^k)}{P_{0T}F_0(T)}
       =\E^{\Q^T}\!\left[(e^{X_T}-e^k)^+\right].             \label{eq:normcall}
\end{equation}
We assume deterministic rates and carry, as in the experiments, so this
normalization yields a common forward-martingale description; stochastic-rate
or hybrid corrections belong in the backbone.
For a smooth price surface, log-strike Breeden--Litzenberger differentiation
\cite{breeden1978} gives
$\Gamma_c:=c_{kk}-c_k=e^k r(T,k)$, with $r(T,\cdot)$ the density of $X_T$. The local-volatility log-forward
\begin{equation}
 \dd X_t=-\tfrac12a_D(t,X_t)\dd t+\sqrt{a_D(t,X_t)}\dd W_t
\end{equation}
has forward equation $\partial_t r=\frac12(\partial_{xx}+\partial_x)(a_Dr)$ and option-price equation
\begin{equation}
 c_T=\tfrac12a_D\Gamma_c,\qquad a_D=\frac{2c_T}{c_{kk}-c_k}.    \label{eq:dupireprice}
\end{equation}
This is the log-forward form of the local-volatility construction
\cite{dupire1994}; in finite data, an arbitrage-aware smoother is
therefore part of the inverse problem \cite{fengler2009}.

Desks quote Black IV rather than $c$. Write $c=\mathcal B(k,T,I)$, $b=I^2$ and $q=\log b$, so total variance is $w=Te^q$. Direct substitution in Gatheral's total-variance formula \cite{gatheral2014} yields
\begin{align}
G_q[q]&=\left(1-\frac{kq_k}{2}\right)^2
 +\frac{Te^q}{2}q_{kk}+\frac{Te^q}{4}q_k^2
 -\frac{T^2e^{2q}}{16}q_k^2,                                  \label{eq:gq}\\
a_D&=\frac{e^q(1+Tq_T)}{G_q[q]},\qquad
\cR_D(q,a)=aG_q[q]-e^q(1+Tq_T).                               \label{eq:dupq}
\end{align}
The division-free residual $\cR_D$ removes algebraic gradient explosions but
does not cure inverse ill-conditioning: when $G_q$ is small it is also weakly
sensitive to $a_D$. We therefore enforce margins $G_q[q]\ge g_{\min}>0$ and
$1+Tq_T\ge\chi_{\min}>0$ on the trusted domain and flag smaller-density regions
as extrapolation; positivity of $I$, $a_D$ and later $m$ is enforced by
exponential or softplus heads. We verified \eqref{eq:gq} symbolically by substituting $w=Te^q$ into
$g(w)=(1-kw_k/(2w))^2-\frac14w_k^2(1/w+1/4)+\frac12w_{kk}$; all $q_k^2$ terms and signs in \eqref{eq:gq} follow from this identity.

\subsection{The SV lift and quotient Fokker--Planck equation}
Let a generic Markovian (or Markovian-lifted) backbone be
\begin{equation}
 \dd Y_t=\mu_\theta(Y_t)\dd t+\Gamma_\theta(Y_t)\dd B_t,
 \quad \dd\langle W,B\rangle_t=\rho_\theta(Y_t)\dd t,
\end{equation}
with $\norm{\rho_\theta(y)}_2\le1$ so that the block covariance is admissible,
and set $V=V_\theta(Y)>0$, $\zeta=\Gamma_\theta\rho_\theta$.
Given leverage $\ell>0$,
\begin{equation}
 \dd X_t=-\tfrac12\ell^2(t,X_t)V(Y_t)\dd t
          +\ell(t,X_t)\sqrt{V(Y_t)}\dd W_t.                    \label{eq:lsv}
\end{equation}
For joint density $p(t,x,y)$ define $r_p=\int p\dd y$, $M_p=\int V(y)p\dd y$ and $m_p=M_p/r_p$. Integrating the fixed-$\ell$ forward equation over $y$ shows that the $X$-marginal has flux $\ell^2M_p$. Matching \eqref{eq:dupireprice} therefore requires Gy\"ongy's identity \cite{gyongy1986}
\begin{equation}
 \ell^2(t,x)m_p(t,x)=a_D(t,x).                                 \label{eq:projection}
\end{equation}
Because $m_p$ is computed under dynamics already containing $\ell$, this is a fixed point in laws, not a pointwise regression.

Substituting $\ell^2=a_Dr_p/M_p$ into the adjoint of the generator of \eqref{eq:lsv} gives the central nonlinear equation
\begin{align}
\partial_t p={}&\tfrac12(\partial_{xx}+\partial_x)
 \left[a_D\frac{r_p}{M_p}Vp\right]-\nabla_y\!\cdot(\mu_\theta p)\notag\\
&+\tfrac12\nabla^2_{yy}\!:\!\left[(\Gamma_\theta\Gamma_\theta^\top)p\right]
 +\partial_x\nabla_y\!\cdot\left[
 \sqrt{a_D\frac{r_p}{M_p}}\sqrt V\,\zeta p\right],             \label{eq:nonlinearFP}
\end{align}
with initial condition $p(0,x,y)=\delta_0(x)\nu_0^\theta(y)$.
The sign and coefficient of the mixed term follow from the cross-covariance
$\dd\langle X,Y\rangle_t=\ell\sqrt V\,\zeta\dd t$; the adjoint has a positive $\partial_x\nabla_y$ divergence. Under the imposed decay or no-flux boundary conditions, integrating \eqref{eq:nonlinearFP} in $y$ cancels all $y$-divergences and leaves exactly $\partial_t r_p=\frac12(\partial_{xx}+\partial_x)(a_Dr_p)$. Hence uniqueness of the marginal equation implies vanilla repricing.

The backward audit is equally important. For $\psi_k(x)=(e^x-e^k)^+$, $u_k(t,x,y;T)=\E[\psi_k(X_T)\mid X_t=x,Y_t=y]$ solves $u_t+\cL_{\ell,\theta}u=0$ with terminal value $\psi_k$. Forward--backward duality requires
\begin{equation}
 c(T,k)=\iint\psi_k(x)p(T,x,y)\dd x\dd y
 =\int u_k(0,0,y;T)\nu_0^\theta(\dd y).                       \label{eq:duality}
\end{equation}

\section{Residual-coupled neural operator}
Given an unordered quote book $Q$ with bid--ask weights and an SV backbone
$\theta$, the model learns the conditional map
$(Q,\theta,T,k)\mapsto(q,a_D,m,\ell^2)$, with $I=e^{q/2}$. We implement this map
with either a DeepONet branch--trunk decoder or an FNO; optional density and
backward-price heads provide witnesses for the full residual audit.
Section~4.1 describes the architecture and residual-based training.
Section~4.2 gives the structural basis for joint prediction and annealed
teacher supervision, together with the conditional consistency result.
Section~4.3 covers the pricing-stack interface and inference-time diagnostics.

\subsection{Operator construction and training}
\subsubsection{Conditional operator architecture and complexity}
Each normalized quote $z_i$ records maturity $T_i$, log-moneyness $k_i$,
bid and ask prices, and a confidence weight $\omega_i$.
A permutation-invariant Deep-Sets branch representation \cite{zaheer2017} is
\begin{equation}
\begin{aligned}
s_Q&=\sum_{i=1}^{n_q}\alpha_i\phi_Q(z_i),\qquad
\alpha_i=\omega_i\big/\!\sum_j\omega_j,\\[-2pt]
h_Q&=\rho_Q\!\left(s_Q,\log\!\sum_i\omega_i,n_{\rm act}\right),
\qquad h_\theta=\eta_\theta(\theta),
\end{aligned}                                                  \label{eq:encoder}
\end{equation}
with masks entering $\phi_Q$ when quote grids are irregular.  Concatenating
$h=(h_Q,h_\theta)$ makes the output conditional on both the observed marginal
information and the chosen SV fibre.  For field
$r\in\{q,a_D,m,\ell^2\}$, a branch--trunk realization is
\begin{equation}
\widetilde r(T,k;Q,\theta)
=B_r(h)^\top\tau_r(T,k)+\beta_r(T,k),                         \label{eq:deeponet}
\end{equation}
followed by $q=\widetilde q$ and positive transforms
$a_D=\epsilon_a+\operatorname{softplus}(\widetilde a)$,
$m=\epsilon_m+\operatorname{softplus}(\widetilde m)$ and
$\ell^2=\epsilon_\ell+\operatorname{softplus}(\widetilde\ell)$.
The log-variance coordinate then gives $I=e^{q/2}$ exactly.  Sharing the
branch state across heads makes smile level, skew and curvature available to
the projection heads, while separate trunks preserve field-specific spatial
resolution \cite{lu2021}. A gridded alternative replaces \eqref{eq:deeponet}
with spectral convolutions \cite{li2021} and retains the same heads and residuals; this is the FNO branch
used in the 500-state capacity experiment.

The branch--trunk decoder represents a function family that can be queried on
PDE, quote or stress grids without retraining. Its $C^2$ coordinate activations
provide $q_T,q_k,q_{kk}$ through $\tau_q$, allowing Dupire constraints between
quotes; ReLU would give $q_{kk}=0$ almost everywhere. The gridded FNO uses
spectral differentiation and interpolation but lacks this arbitrary-node
query guarantee. For $M$ nodes, operator inference costs
$O(n_qd_b+Md_t)$ and is parallel, whereas particle calibration requires at
least $O(Nn_t)$ simulation plus sequential regression and interpolation.
Thus the stochastic fixed point is learned offline and evaluated online by
\eqref{eq:encoder}--\eqref{eq:deeponet}.

\subsubsection{Coupled residual objective}
Let
$d_i=\operatorname{dist}(\mathcal B(k_i,T_i,I_i),[c_i^{\rm bid},c_i^{\rm ask}])$.
Set $c_\phi=\mathcal B(k,T,e^{q/2})$ and recall
$r_p=\int p\dd y$, $M_p=\int Vp\dd y$. For predicted fields (hats omitted),
the fully identifying population residual system is
\begin{align}
J_{\rm data}&=\E\sum_i\omega_i d_i^2+\lambda_C\E\norm{c_\phi-\Pi_C(Q)}_{L^2(\nu_{\cD})}^2,\notag\\
J_{\rm arb}&=\E\{\norm{\pos{g_{\min}-G_q[q]}}^2+
 \norm{\pos{\chi_{\min}-(1+Tq_T)}}^2\},\notag\\
J_{\rm Dup}&=\E\norm{a_DG_q[q]-e^q(1+Tq_T)}^2,\notag\\
J_{\rm proj}&=\E\norm{\ell^2m-a_D}^2,\qquad
J_{\rm mom}=\E\norm{mr_p-M_p}^2,                              \label{eq:losses}\\
J_{\rm fwd}&=\E\!\left[\norm{\partial_t p-\mathcal F_{a_D,\theta}(p)}_{W^{\prime}}^2+
 \norm{p(0)-\delta_0\!\otimes\!\nu_0^\theta}_{H^{-s}}^2\right.\notag\\[-2pt]
&\hspace{41mm}\left.+\cR_{\rm mass}^2+\cR_{\partial p}^2\right],\notag\\
J_{\rm bwd}&=\E\!\left[\norm{u_t+\cL_{\ell,\theta}u}^2+
 \norm{u(T)-\psi_k}^2+\cR_{\partial u}^2\right],\notag\\
J_{\rm dual}&=\E\left|\int u(0,0,y)\nu_0^\theta(\dd y)
 -\iint\psi_k(x)p(T,x,y)\dd x\dd y\right|^2.\notag
\end{align}
Here $\mathcal F_{a_D,\theta}$ is the right-hand side of
\eqref{eq:nonlinearFP}, $\cR_{\rm mass}$ is the $L_t^2$ mass error, and the
boundary terms impose the stated decay, zero-flux and pricing conditions.
Finite bid--ask observations do not identify the off-grid chart, so
Theorem~\ref{thm:consistency} requires $\Pi_C$,
supplied either as a cleaner input or as collocation supervision. The fast
models use the data, arbitrage, Dupire and projection terms; the remaining
losses define the stronger witness system and are not evaluated as an
empirical ablation. Automatic differentiation supplies $q_T,q_k,q_{kk}$,
making $J_{\rm Dup}$ a Sobolev loss; explicit Greek labels may also be added
\cite{huge2020}.

\subsubsection{Residual normalization and curriculum}
The losses in \eqref{eq:losses} have different units and effective sample
sizes, so fixed raw weights are not meaningful.  For residual family $r$, let
$s_r$ be a robust scale estimated on the current curriculum stage (median
absolute deviation for pointwise residuals and a seed variance for particle
views).  The optimized quantities are
\begin{equation}
\bar J_r=\frac{J_r}{s_r^2+\varepsilon},\qquad
J_{\rm agg}=\frac1\alpha\log\sum_{r\in\mathcal T}
\pi_r\exp(\alpha\bar J_r),                                   \label{eq:normalizedloss}
\end{equation}
with $\sum_r\pi_r=1$.  Scales are stopped-gradient statistics: otherwise a
head could reduce its weight by inflating its own normalization.  The smooth
maximum in \eqref{eq:normalizedloss} concentrates updates on the least
satisfied identity while remaining differentiable; clipping its exponent
prevents one boundary outlier from dominating a minibatch.

Training proceeds in three stages. The chart stage learns Black inversion and
quote interpolation on collocation points stratified by maturity, moneyness and
the bid--ask mask. The generator stage activates calendar, butterfly and
Dupire residuals, oversampling short maturities and trusted-interior edges. The
projection stage introduces SV backbones and moment/leverage heads, with
optional forward/backward witnesses for the full audit. The curriculum moves
from constant variance through Heston-type states to rough or multi-factor
teachers. Equation~\eqref{eq:teacher} down-weights noisy or inconsistent
teachers, and annealing begins after validation residuals stop improving.
Market-state minibatches keep derivatives and integral constraints coherent;
optional density and backward heads are discarded at inference.

\subsection{Structural justification and conditional consistency}
The residual system must constrain the calibration fields between quote nodes,
not merely reproduce the observed prices. Two design choices are central: the
fields share one conditional representation, and teacher labels are annealed
rather than imposed as exact truth. After establishing these points, we state
the assumptions under which vanishing residuals identify the target operator.

\subsubsection{Joint learning of the calibration fields}
A leverage-only target hides which part of a repricing error comes from the marginal generator and which from the conditional moment.  Let $(a_D,m,\ell)$ be a consistent triple and let hats denote predictions.  Introducing the predicted projection residual
$\widehat R_{\rm proj}=\widehat\ell^{\,2}\widehat m-\widehat a_D$ gives the exact quotient identity
\begin{align}
\widehat\ell^{\,2}-\ell^2
&=\frac{\widehat a_D-a_D}{\widehat m}
-\frac{a_D(\widehat m-m)}{m\widehat m}
+\frac{\widehat R_{\rm proj}}{\widehat m}.                    \label{eq:decomp}
\end{align}
Consequently, if $m,\widehat m\ge m_{\min}>0$ and the local-variance,
moment and projection errors are bounded by $\delta_a,\delta_m,\delta_p$,
\begin{equation}
\norm{\widehat\ell^{\,2}-\ell^2}\le
\frac{\delta_a}{m_{\min}}+
\frac{\norm{a_D}_\infty\delta_m}{m_{\min}^2}+
\frac{\delta_p}{m_{\min}}.                                   \label{eq:jointbound}
\end{equation}
Joint heads and residuals control every term in \eqref{eq:decomp}; pointwise
imitation of a stochastic leverage label controls none of the first three.
The shared encoder also has a statistical role: IV level, skew, curvature and
term structure inform $a_D$, while the same latent market state conditions how
that marginal is embedded in the SV fibre.

\begin{mainproposition}[Truncated-domain price stability]
Fix the backbone $\theta$ and a bounded, variance-truncated domain $\Omega$.
Let $u$ and $\widehat u$ be the backward prices generated by the true and
predicted LSV coefficients. Assume that both generators are uniformly
parabolic on $\Omega$, their coefficients lie in $W^{1,\infty}(\Omega)$,
$m,\widehat m\ge m_{\min}>0$, and
$\ell,\widehat\ell\ge\ell_{\min}>0$. For every Lipschitz terminal payoff with
compatible boundary data, there is a constant $C_T$ such that
\begin{equation}
|\widehat u(0)-u(0)|
\le C_T\!\left(
\norm{\widehat a_D-a_D}
+\norm{\widehat m-m}
+\norm{\widehat R_{\rm proj}}\right).                         \label{eq:pricestability}
\end{equation}
Here the norms are $L^\infty(\Omega)$ norms, and $C_T$ depends only on
$T,\Omega$, the parabolicity constants and the coefficient bounds.
\end{mainproposition}

\begin{proof}[Proof sketch]
Let $\Delta$ denote the sum on the right of \eqref{eq:pricestability}.
By \eqref{eq:decomp} and the lower bound on $m$,
$\norm{\widehat\ell^{\,2}-\ell^2}_\infty\le C\Delta$; the lower bound on
leverage then gives $\norm{\widehat\ell-\ell}_\infty\le C\Delta$. Extend the
coefficients boundedly and Lipschitz-continuously and couple
$Z=(X,Y)$ and $\widehat Z=(\widehat X,\widehat Y)$ with the same Brownian
motion. Burkholder--Davis--Gundy and Gronwall estimates yield
$\E\sup_{t\le T}|\widehat Z_t-Z_t|^2\le C_T\Delta^2$.
The Lipschitz payoff bound gives \eqref{eq:pricestability}.
\end{proof}
Genuine Heston degeneracy at $V=0$, barrier corners and discontinuous payoffs
are outside the proposition and require independent grid-convergence audits.
The constant may also deteriorate at long horizons, so the diagnostic below
is not a universal confidence interval.

\subsubsection{Annealed teacher supervision}
Finite quotes leave derivative information underdetermined. The following
result makes this non-identification explicit.

\begin{mainproposition}[Finite quotes do not identify Dupire]
Let $q\in C^{1,2}(\cD)$ be strictly arbitrage-admissible and
$G=\{(T_i,k_i)\}_{i=1}^n$. For every interior
$z_0=(T_0,k_0)\in\cD\setminus G$, there exists
$\psi\in C_c^\infty(\cD)$ such that $\psi|_G=0$ and
$D'(q)[\psi](z_0)\ne0$, where
$D(q)=e^q(1+Tq_T)/G_q[q]$.
\end{mainproposition}

\begin{proof}
Choose $\varrho\in C_c^\infty(\cD)$ with
$\operatorname{supp}\varrho\cap G=\varnothing$ and $\varrho(z_0)=1$; set
$\psi(T,k)=(T-T_0)\varrho(T,k)$. At $z_0$,
$\psi=\psi_k=\psi_{kk}=0$ and $\psi_T=1$, so the directional derivative of
$G_q[q]$ vanishes while that of $e^q(1+Tq_T)$ equals $e^{q(z_0)}T_0$.
Hence
$D'(q)[\psi](z_0)=e^{q(z_0)}T_0/G_q[q](z_0)\ne0$.
Strict admissibility is open in $C^{1,2}(\cD)$, so $q+\epsilon\psi$ remains
admissible for sufficiently small $|\epsilon|$.
\end{proof}

Accordingly, SSVI, Andreasen--Huge and particle outputs are treated as
heterogeneous instruments. For head $h$, model teacher $j$ as
$h_j^T=h^\star+b_j+\varepsilon_j$, with $\E\varepsilon_j=0$ and
$\operatorname{Var}\varepsilon_j=s_j^2$, and assign the weight
\begin{equation}
 \lambda_j(z,\tau)=\lambda_j^{(0)}(\tau)
 \frac{\exp[-\gamma|\cR(h_j^T)(z)|]}{s_j^2(z)+\epsilon},       \label{eq:teacher}
\end{equation}
where $\lambda_j^{(0)}(\tau)\downarrow0$ as clean labels and residual evidence
take over. Under local coercivity,
$J_{\rm res}(h)-J_{\rm res}(h^\star)\ge
\mu\norm{h-h^\star}^2$, minimizers of $J_{\rm res}+\eta J_T$ lie
$O(\sqrt\eta)$ from $h^\star$; a nonvanishing weight on a biased teacher instead
retains weighted-average bias. This explains why teachers stabilize
scarce-label training, whereas teacher-heavy full-label training inherits SSVI
bias.

Particle teachers have structured noise. For particles $(X_t^i,V_t^i)$,
$i=1,\ldots,N$, bandwidth $h$ and kernel $K_h$, the conditional-moment estimate
is
\begin{equation}
\widehat m_{N,h,\Delta t}(t,x)=
\frac{\sum_iV_t^iK_h(X_t^i-x)}{\sum_iK_h(X_t^i-x)},\qquad
\widehat\ell^{\,2}=\frac{a_D}{\widehat m_{N,h,\Delta t}}.       \label{eq:kernel}
\end{equation}
In one conditioning dimension, a second-order kernel has the indicative
error balance
\begin{equation}
\E(\widehat m-m)^2\approx C_1h^4+\frac{C_2}{Nh}
+C_3\Delta t^\beta+C_4\Delta x^\gamma,                        \label{eq:kernelmse}
\end{equation}
so $h\asymp N^{-1/5}$ balances smoothing bias and variance. Time stepping,
interpolation, boundaries and the quotient in \eqref{eq:kernel} add
heteroskedasticity. A coarse run therefore supplies useful geometry and a warm
start rather than noiseless pointwise truth; amortization averages this noise
across market states.

\begin{figure}[htbp]
 \centering
 \includegraphics[width=\columnwidth]{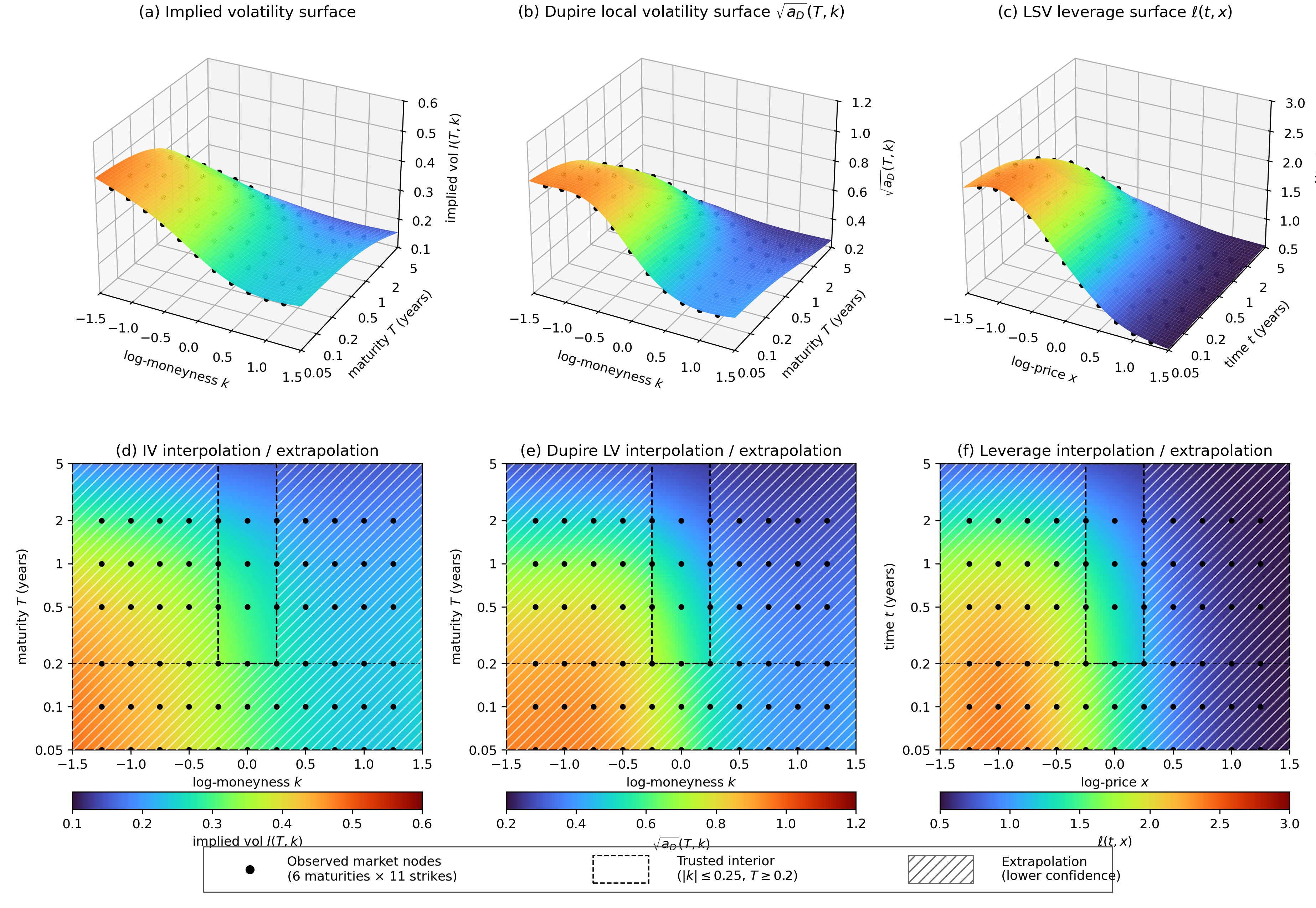}
 \caption{Representative calibration triple. Black points mark quote nodes, the
 dashed box the trusted interior, and hatching lower-confidence extrapolation;
 colors denote field values, not errors.}
 \label{fig:surfaces}
\end{figure}

Figure~\ref{fig:surfaces} is a qualitative diagnostic of quote coverage and
extrapolation, not an accuracy plot.

\subsubsection{Conditional consistency}
The global well-posedness of calibrated LSV models is delicate \cite{abergel2010,lacker2020,bayer2024}. We therefore state the conditional guarantee needed for amortization.

\begin{maintheorem}[Conditional identification and amortized convergence]
\label{thm:consistency}
Let $\Xi$ be the set of admissible inputs $\xi=(Q,\theta)$, endowed with a
probability law $\Pi$, and fix a cleaning map $\Pi_C$. For each $\xi$, let
$H_\xi^\star=(q,a_D,p,m,\ell^2,u)_\xi^\star$ be the full solution selected by
$(\Pi_C(Q),\theta)$, and let $\mathcal R_\xi$ denote \eqref{eq:losses}. Assume
that the selected charts are compact in $C^{1,2}(\cD)$; Black vega,
$G_q[q]$, $1+Tq_T$, $r_p$ and $m$ are uniformly positive on the trusted
domain; and the forward and backward problems are uniquely solvable. Suppose
also that the local coercivity estimate
$\norm{H-H_\xi^\star}_{\mathcal X}\le
C\norm{\mathcal R_\xi(H)}_{\mathcal Y}$ holds uniformly for
$H\in\mathcal U_\xi$ on admissible neighbourhoods $\mathcal U_\xi$. Let
$\mathfrak H_n$ be
graph-norm dense classes taking values in $\mathcal U_\xi$, and assume a
uniform law of large numbers for their empirical residual risks. Define
$\mathcal J(\mathcal H)=
\E_\Pi\norm{\mathcal R_\xi(\mathcal H(\xi))}_{\mathcal Y}^2$.
Then: (i) $\mathcal R_\xi(H)=0$ implies $H=H_\xi^\star$;
(ii) $\inf_{\mathcal H\in\mathfrak H_n}\mathcal J(\mathcal H)\to0$; and
(iii) any empirical near-minimizer $\widehat{\mathcal H}_{n,N}$ with
vanishing sampling and optimization gaps satisfies
\begin{equation}
 \norm{\widehat{\mathcal H}_{n,N}-\mathcal H^\star}
 _{L^2(\Pi;\mathcal X)}
 \xrightarrow{\mathbb P}0.                                  \label{eq:operatorconsistency}
\end{equation}
Consequently, the associated calibration triple converges to
$\mathcal G^\star$ in $L^2(\Pi;C(\cD))$.
\end{maintheorem}

\begin{proof}[Proof sketch]
Once $\Pi_C$ selects $c_\xi^\star$, the positive-vega bound makes Black
inversion pointwise injective and identifies its continuous chart
$I_\xi^\star$, hence $q_\xi^\star=2\log I_\xi^\star$. Finite bid--ask
observations alone do not provide this identification, which is why the
cleaning convention is fixed. Equation~\eqref{eq:dupq} and the lower bound on
$G_q$ then identify $a_{D,\xi}^\star$. The forward initial, mass and boundary
residuals, together with uniqueness of \eqref{eq:nonlinearFP}, identify
$p_\xi^\star$. Since $r_{p_\xi^\star}$ is bounded away from zero,
$J_{\rm mom}=0$ gives
$m_\xi^\star=M_{p_\xi^\star}/r_{p_\xi^\star}$; positivity of
$m_\xi^\star$ then identifies
$(\ell_\xi^\star)^2=a_{D,\xi}^\star/m_\xi^\star$. The terminal and backward
residuals identify $u_\xi^\star$, while the dual residual gives an independent
forward--backward audit. This proves (i).

Graph-norm density and residual continuity give (ii). The local quotient
estimate
$|a/m-a'/m'|\le |a-a'|/m_{\min}
+\norm{a'}_\infty|m-m'|/m_{\min}^2$
controls the leverage component in the same graph norm. The uniform law of
large numbers and the vanishing optimization gap imply
$\mathcal J(\widehat{\mathcal H}_{n,N})\to0$ in probability. Squaring the
coercivity estimate and integrating over $\Pi$ then yields (iii)
\cite{vandervaart1998}.
\end{proof}

Theorem~\ref{thm:consistency} clarifies the role
of the residual system. Under the stated well-posedness and stability
assumptions, zero residual identifies the cleaned calibration triple, and
vanishing empirical residual risk implies convergence of the learned operator.
The residuals therefore couple the output fields during training and provide
a consistency check at inference.

\subsection{Deployment and residual diagnostics}
Deployment uses the predicted fields in two places. They enter the desk's
existing pricing engines and feed the residual checks that determine whether
the result can be used directly or should return to a conventional calibration.

\subsubsection{Integration with an exotic-pricing stack}
The learned operator removes the calibration fixed point, not the payoff solve.
For a European or barrier payoff $\Psi$, a conventional finite-difference
engine solves
\begin{equation}
u_t+\cL_{\ell,\theta}u=0,\qquad u(T,x,y)=\Psi(x,y),             \label{eq:exoticpde}
\end{equation}
with an absorbing/rebate boundary for barriers.  Heston-type LSV produces the
usual two-dimensional PDE with a correlation mixed derivative;
alternating-direction implicit (ADI) splitting, Rannacher smoothing and
monotone interpolation of $\ell$ can remain unchanged.
For an arithmetic running state $A_t=\int_0^t e^{X_s}\dd s$, the augmented
generator is $\cL_{\ell,\theta}+e^x\partial_A$; cliquets and autocallables add
observation-date jump conditions. Finite-difference (FD), ADI or Monte Carlo
(MC) solvers then give the online chain
\[
(Q,\theta)\xrightarrow{\;\mathcal G_\phi\;}(I,a_D,m,\ell)
\xrightarrow{\;\text{FD/ADI or MC}\;}\text{price and Greeks}.
\]
Intraday scenarios use the residual-certified output; other workflows may
warm-start a short particle polish and retain projection and repricing
residuals. The same interface supports hybrid LSV stacks and particle control
variates \cite{cozma2019}, while payoff-engine safeguards remain independent.

\subsubsection{Residual-based diagnostics and governed inference}
Low latency is useful only if a desk can decide when the amortized answer is
safe to consume. Each inference therefore carries the diagnostic vector
\begin{equation}
\mathfrak C(Q,\theta)=
\big(r_{\rm quote},\,\underline g,\,\underline\chi,\,
r_{\rm Dup},\,r_{\rm proj},\,r_{\rm dual}^{\rm audit},\,d_{\rm OOD}\big),  \label{eq:cert}
\end{equation}
where $r_{\rm quote}$ is bid--ask distance, $\underline g=\inf G_q[q]$ is the
butterfly margin, $\underline\chi=\inf(1+Tq_T)$ is the calendar margin, and
$r_{\rm Dup}$ and $r_{\rm proj}$ are normalized residuals. We specify
$d_{\rm OOD}^2=(h-\bar h)^\top(\widehat\Sigma_h+\lambda I)^{-1}(h-\bar h)$
is an out-of-distribution (OOD) distance, with training mean/covariance and
validation-chosen ridge $\lambda$; other scores require fresh validation. The
first two residuals test the marginal chart; the
projection residual tests whether the predicted LSV dynamics can support that
marginal under the selected backbone. The fast path reports these terms;
$r_{\rm dual}^{\rm audit}$ requires the optional backward witness. Teacher seed
dispersion can be appended when particle views are available.

This vector defines three governed modes. \emph{Accept} uses the neural
triple when positivity margins and residual thresholds pass.
\emph{Polish} initializes a short particle solve from $\widehat\ell$ when the
chart is admissible but the projection residual is elevated. \emph{Reject}
returns no certified LSV calibration when no-arbitrage fails,
conditional-moment support is too small, or a polish cannot reduce the
residual. Thresholds are set on a validation distribution and tightened for
contractual valuation. This matters in high-vol-of-vol regimes where a fixed
point may be difficult to certify. For a
validation error $E$ and acceptance region $A_\tau$ induced by thresholds
$\tau$, deployment calibration must report
\begin{equation}
 R(\tau)=\E[E\mid\mathfrak C\in A_\tau],\qquad
 C(\tau)=\Pr(\mathfrak C\in A_\tau),                         \label{eq:selective}
\end{equation}
the selective risk and coverage, with thresholds set on validation data and
frozen for test \cite{angelopoulos2021}. The experiments do not estimate
\eqref{eq:selective}, so $\mathfrak C$ is not a calibrated confidence
certificate; its 0.59 particle-residual correlation is only preliminary ranking
evidence. Calibrated thresholds could minimize
$t_{\rm op}+[1-C(\tau)]t_{\rm polish}$ subject to $R(\tau)\le\varepsilon$;
current data do not support this risk--compute claim.

\begin{figure}[t]
 \centering
 \includegraphics[width=0.80\textwidth]{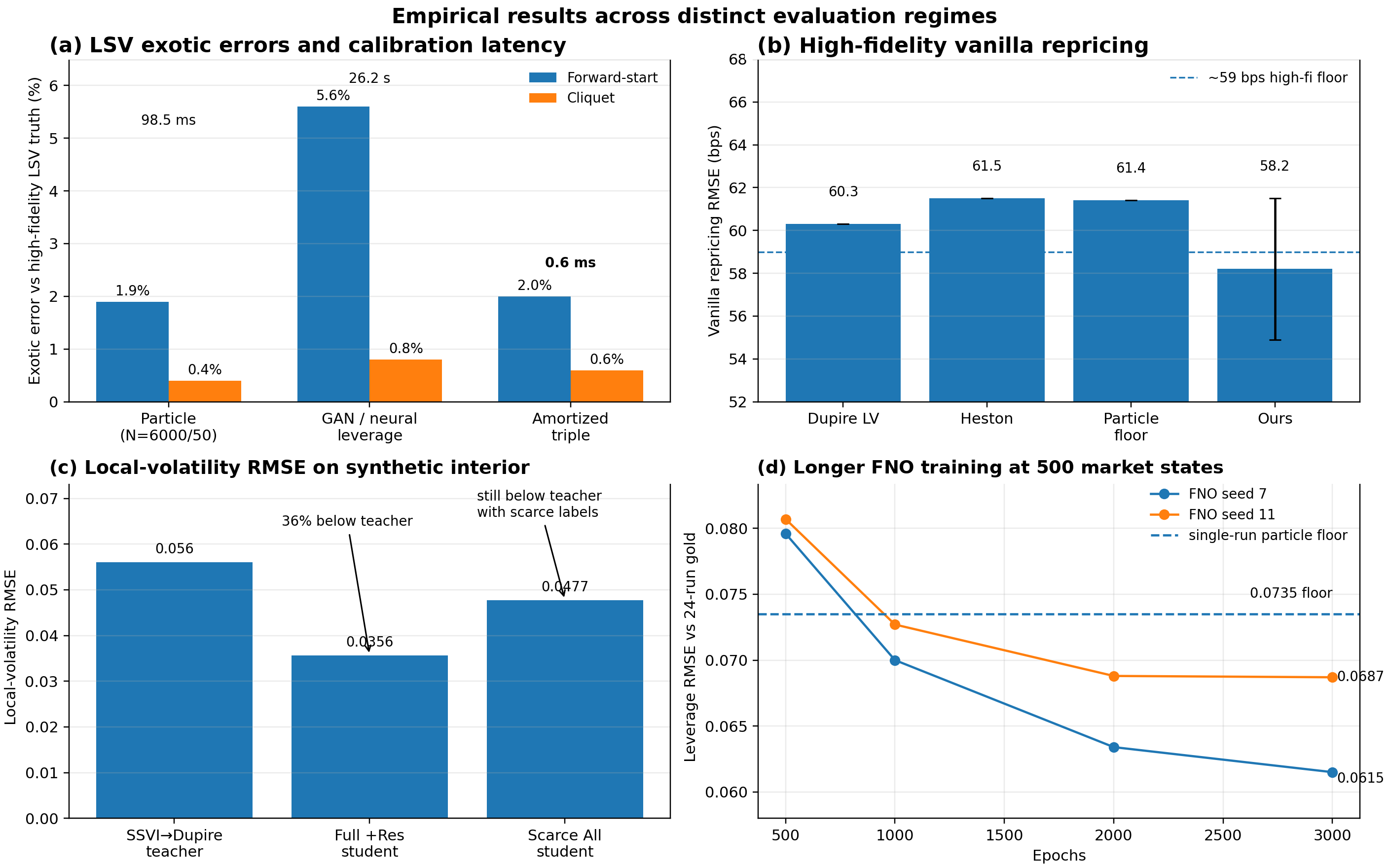}
 \caption{Empirical results from distinct regimes. (a) Exotic-error point
 estimates differ from the particle implementation by 0.1/0.2 percentage
 points; intervals are unavailable. (b) High-fidelity vanilla repricing reaches
 the measured numerical floor. (c) The local-volatility head has 36\% lower
 RMSE than this one-surface SSVI--Dupire estimator on the synthetic interior.
 (d) At
 500 states, two FNO seeds fall below the single-run particle deviation
 benchmark against a 24-run mean.}
 \label{fig:main}
\end{figure}

\begin{figure}[t]
 \centering
 \includegraphics[width=0.75\textwidth]{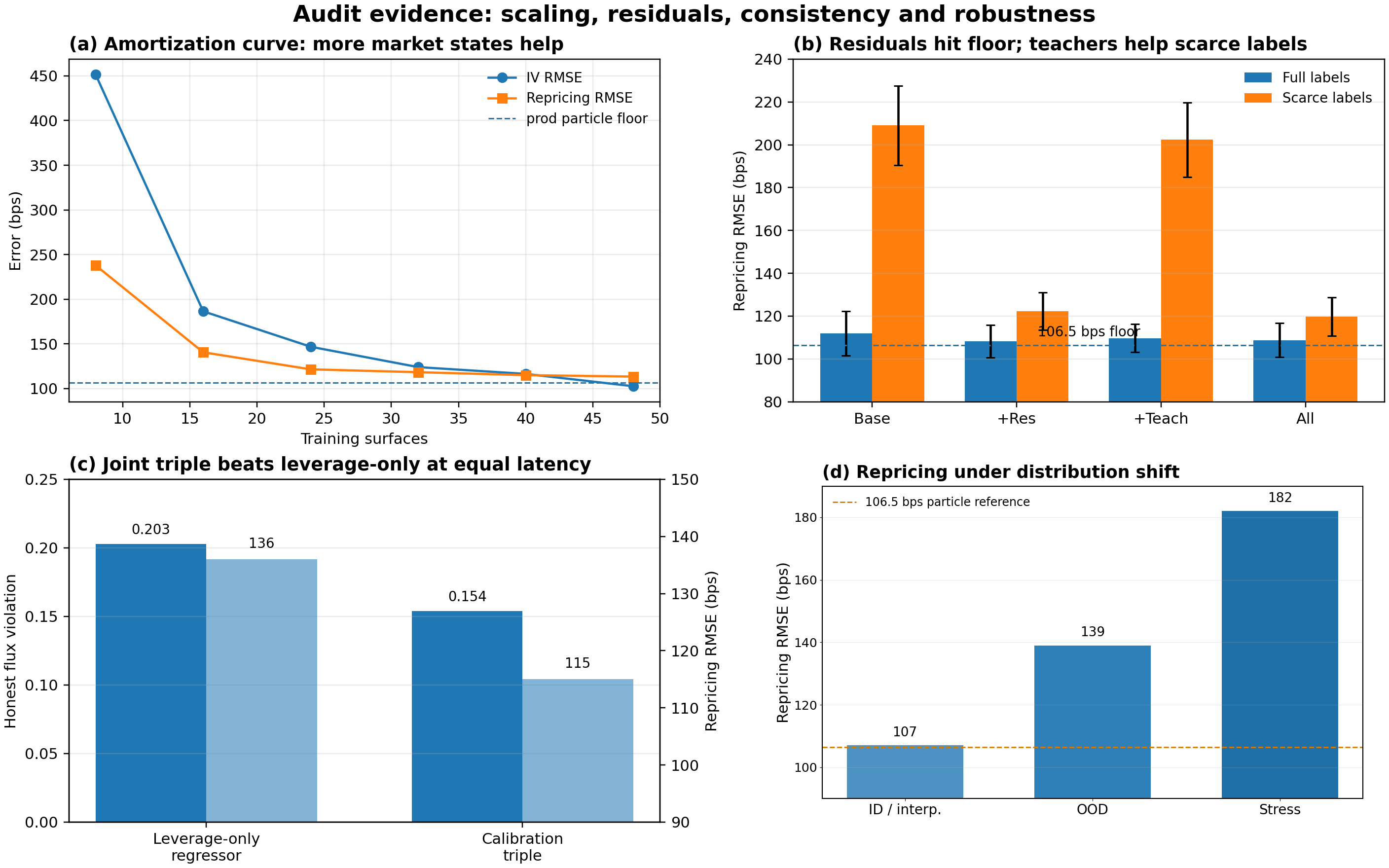}
 \caption{Audit evidence and operational interpretation. (a) Errors decrease
 with market states. (b) Residuals stabilize full-label fits while teachers
 rescue scarce-label training. (c) Joint learning improves independent flux
 and repricing. (d) Repricing under interpolation, OOD and stress shifts on one
 common basis-point axis.
 Error bars are across the reported seeds.}
 \label{fig:audits}
\end{figure}

\section{Experiments}
\subsection{Design and verified numerics}
The main suite uses 48 training and 12 held-out rough-Bergomi market states
\cite{bayer2016}, a Heston-type application backbone \cite{heston1993}, and
quote grids of six maturities by eleven log-moneyness nodes. Metrics are
restricted to $|k|\le0.25$, $T\ge0.2$; wings and shorter maturities are flagged
as extrapolation. Vanilla errors are IV RMSE in basis points; local-volatility
and leverage errors are absolute RMSE in volatility units. We compare
SSVI--Dupire and Heston baselines, a particle fixed point, a per-surface
generative adversarial network (GAN)/neural-leverage method
\cite{cuchiero2020}, a leverage-only network, and ablations without residuals
or teachers.

We verified each analytic and numerical link before fitting. The tests obtained
Black inversion error $3.5\times10^{-16}$, flat-Black Dupire error at machine
precision, convergence between price- and IV-coordinate Dupire from 2.02\% to
0.043\%, rough-Bergomi Volterra covariance error $10^{-16}$, the
constant-backbone identity $\ell=\sqrt{a_D/v_0}$, and Heston
Lewis-versus-Monte-Carlo agreement of 7.6 bps. Projection is scored against
independent targets,
\begin{equation}
 \operatorname{Flux}(\widehat\ell)=
 \frac{\norm{\widehat\ell^{\,2}m^\star-a_D^\star}}{\norm{a_D^\star}},              \label{eq:flux}
\end{equation}
not against the model's own $\widehat m$ and $\widehat a_D$. This avoids an
endogenous score that can remain small when all three predicted fields drift
together.

\subsubsection{Evaluation hierarchy and claim discipline}
Pooling the four evidence regimes would overstate precision.
Table~\ref{tab:protocols} therefore separates the five-seed ablation suite, the
high-fidelity vanilla study, the 500-state noise-averaging experiment and the
common-reference LSV exotic audit by the question each can answer.

\begin{center}
\begin{minipage}{\columnwidth}
\refstepcounter{table}\label{tab:protocols}
\small\textbf{Table \thetable: Evidence regimes have distinct budgets and
inferential roles.}\par
\footnotesize
\setlength{\tabcolsep}{2.0pt}
\begin{tabularx}{\columnwidth}{@{}>{\raggedright\arraybackslash}p{0.16\columnwidth}
>{\raggedright\arraybackslash}p{0.41\columnwidth}
>{\raggedright\arraybackslash}X@{}}
\toprule
Regime & Evidence budget & Supported question\\
\midrule
Main & 48/12; 5 seeds; reference particle & ablation, arbitrage, scaling, OOD\\
High fidelity & 28/8; 2 seeds; $N=40{,}000$/120 & vanilla numerical floor\\
Noise mean & 500 states; 2 FNO seeds; 24 runs & field denoising\\
Exotics & held out; truth $N=60{,}000$/150 & exotic point errors\\
Online & particle $N=6{,}000$/50; GAN 150/4,000 & accuracy--latency frontier\\
\bottomrule
\end{tabularx}
\end{minipage}
\end{center}

Splits are by market state, not $(T,k)$ node, and the trusted interior is fixed
before comparison. Timing excludes shared quote ingestion, includes the
particle method's 0.53 ms Dupire input, and reports batching regimes
separately. The common exotic reference isolates LSV dynamics from vanilla fit.

Uncertainty claims follow the available replication. The
$58.2\pm3.3$ bps result is a descriptive spread across two training seeds, not a
confidence interval, and is \emph{floor-level} rather than superior to the
roughly 59 bps numerical floor. Likewise, the two-seed FNO result is evidence
of noise averaging rather than a population statement, and \emph{zero
violations} means none observed on the evaluated grids and seeds.

\subsection{Accuracy, noise averaging and latency}
Table~\ref{tab:score} and Fig.~\ref{fig:main} collect the principal results.
Against high-fidelity LSV truth ($N=60{,}000$, 150 steps), forward-start and
cliquet point errors differ from the reference particle implementation by
0.1 and 0.2 percentage points; without Monte-Carlo intervals, this is not an
equivalence claim. The reported 150-iteration GAN is slower and less accurate
in this comparison. With $N=40{,}000$, 120-step labels and 28/8 train/test
surfaces, vanilla repricing is $58.2\pm3.3$ bps, versus 60.3 for Dupire, 61.5
for Heston and a 59 bps particle self-repricing floor. The 1.8 bps difference
from that floor is not statistically meaningful.

\begin{center}
\begin{minipage}{\columnwidth}
\centering
\refstepcounter{table}
\label{tab:score}
{Table \thetable: Held-out scorecard (lower is better). The
comparison method is row-specific; the leverage reference is a 24-run
mean.}
\par
\small
\setlength{\tabcolsep}{2.0pt}
\begin{tabular}{@{}lrr@{}}
\toprule
Metric & Competitor & Triple\\
\midrule
Vanilla RMSE (bps) & 60.3 local vol. / 61.5 Heston & $58.2\pm3.3$\\
IV RMSE (bps) & 227.6 SSVI & 41.5\\
Local-volatility RMSE & 0.056 teacher & 0.0356\\
Fwd-start / cliquet error & 1.9\% / 0.4\% particle & 2.0\% / 0.6\%\\
Leverage RMSE & 0.0735 single run & 0.0615--0.0687\\
Calibration time (ms) & 98.5 particle & 0.6\\
Independent flux \eqref{eq:flux} & 0.203 leverage-only & 0.154\\
\bottomrule
\end{tabular}
\end{minipage}
\end{center}

On the rough-Bergomi interior, the IV head gives 41.5 bps error versus 227.6
for SSVI, and local-volatility RMSE falls from 0.056 for SSVI--Dupire to 0.0356
(36\%). The target is the exact Dupire identity; the benchmark is a noisy
finite-data estimator. At 500 states, two FNO seeds give leverage RMSE
0.0615/0.0687 against a 24-run mean, 7--16\% below the 0.0735 single-run
particle deviation. This is noise-averaging evidence, not a population result.
From 8 to 48 training states, IV error decreases
$451\to186\to147\to124\to116\to103$ bps and repricing $238\to113$ bps.

At equal 1.02 ms latency, joint learning improves leverage RMSE, independent
flux and repricing from 0.127/0.203/136 bps to 0.098/0.154/115 bps. With scarce
labels, teachers prevent collapse and give local-volatility/leverage RMSE
0.0477/0.1501; with full labels, teacher-heavy variants inherit about 233 bps
of SSVI bias. Residual training in the latter regime gives IV
$102.7\pm0.6$ bps and repricing $108.3\pm7.6$ bps. The 10.5\%
SSVI-versus-market local-volatility commutator is consistent with
Proposition~2.

The paired 0.60/98.5 ms operator/particle timings use the same hardware and
standard configuration. Table~\ref{tab:latency} separates them from the matched
ablation (1.02 ms) and 984-node CPU checkpoints; particle totals include the
0.53 ms Dupire step. A shared 512 ms two-dimensional ADI exotic solve
reduces the complete-workflow speedup to 1.15--2.53$\times$. Writing total costs as
$C_{\rm op}(M)=T_{\rm train}^{\rm op}+Mt_{\rm op}$ and
$C_{\rm base}(M)=T_{\rm setup}^{\rm base}+Mt_{\rm base}$, a favourable crossing
exists only when $t_{\rm base}>t_{\rm op}$, after
$\lceil(T_{\rm train}^{\rm op}-T_{\rm setup}^{\rm base})_+/
(t_{\rm base}-t_{\rm op})\rceil$ surfaces. Matched offline costs were not
recorded, so this is an online frontier, not a total-cost dominance claim.

\begin{table}[t]
\caption{Timing registry. Paired E timings use the same hardware and standard
configuration; other tags are distinct checkpoints. L uses 984 nodes.}
\label{tab:latency}
\centering
\footnotesize
\setlength{\tabcolsep}{1.5pt}
\begin{tabular}{@{}cllr@{}}
\toprule
Tag & Component & Setting & ms/surface\\
\midrule
E & triple / particle & exotic / $N=6{,}000$, 50 & 0.60 / 98.5\\
A & triple / lev.-only & matched & 1.02 / 1.02\\
L & Dupire FD & 984 nodes & 0.53\\
L & particle & $2^{10}/2^{12}/2^{13}$; 64 steps & 77.93/385.43/783.93\\
L & operator & batch 1/16/64 & 2.37/0.19/0.074\\
P & 2D ADI exotic & $96\times48\times128$ & 512\\
\bottomrule
\end{tabular}
\end{table}

\subsection{Robustness and limitations}
Across five seeds, residual training produced no observed calendar or butterfly
violations on the evaluated grids, versus 2.9\% butterfly violations without
residuals. Repricing error is 107/139/182 bps on interpolation/OOD/stress boxes,
and the projection residual correlates 0.59 with the particle residual in the
volatility-of-volatility sweep. The log-variance head reduces IV error from
136.4 to 102.7 bps (about 25\%); without architecture changes, a two-factor
lognormal Ornstein--Uhlenbeck variance backbone reaches 102--104 bps against an
85 bps numerical floor.

The evidence remains prototype and predominantly synthetic: the principal
suite has 12 test surfaces, the high-fidelity suite eight with two seeds, and
the 500-state FNO study two seeds; exotic estimates retain Monte-Carlo
uncertainty. Boundary/short-maturity error is 175 versus 41.5 bps in the trusted
interior. Pure local volatility may also win a low-budget vanilla metric, which
does not test full LSV dynamics. Production validation requires rolling
historical books, more seeds and states, exotic confidence intervals, hardware
latency percentiles and targeted high-vol-of-vol failures. The present results
therefore support the mechanism and measured ordering, not universal dominance.

\FloatBarrier
These trends have not been tested on historical books. With broader data,
errors should continue to fall as independent market states are added,
residual scores should retain their ranking of particle-polish difficulty, and
the joint-head and teacher-weight effects should persist. If they do not, the
corresponding empirical claims should be narrowed.

\section{Conclusion}
We replace the sequential LSV calibration loop with a projection-consistent
operator that predicts implied volatility, Dupire local variance, the
conditional moment and leverage from finite quotes and an SV backbone. On
synthetic data, the operator gives exotic point errors close to one particle
implementation, improves on two field-estimation benchmarks and reduces
calibration latency to below one millisecond. Future work will test the method
on rolling historical quote books, quantify uncertainty in exotic prices and
compare end-to-end costs under matched settings.

\bibliographystyle{plain}
\bibliography{references}

\end{document}